\documentclass{article}
\usepackage{amssymb,amsmath,graphicx,ucs,hyperref}
\usepackage[utf8x]{inputenc}

\newtheorem{theorem}{Theorem}

\newtheorem{remark}{Remark}
\newtheorem{proposition}{Proposition}
\newtheorem{lemma}{Lemma}
\newenvironment{proof}{\noindent \emph{Proof. }}{\hfill \hbox{\rlap{$\sqcap$}$\sqcup$}\\}

\title{Density of Binary Disc Packings:\\ Playing with Stoichiometry}
\author{Thomas Fernique\\{\small LIPN, Univ. Paris Nord \& CNRS}}
\date{}

\begin{document}

\maketitle

\begin{abstract}
We consider hard-disc mixtures with disc sizes within ratio $\sqrt{2}-1$, that is, the small disc exactly fits in the hole between four large discs.
For each prescribed stoichiometry of large and small discs, the densest packings are rigorously determined via a computer-assisted proof.
The density is maximal for the 1:1 stoichiometry: the large discs then form a square grid in each interstitial site of which a small disc nests.
When there is an excess of large discs, the densest packings are made of a single phase which mixes the two types of discs in a chaotic way (it can be described by square-triangle tilings).
When there is an excess of small discs, on the contrary, a phenomenon of phase separation appears: the large discs are involved in the densest 1:1 stoichiometry phases while the excess of small discs form compact hexagonal phases.
\end{abstract}

%%%%%%%%%%%%%%%%%%%%%
\section{Introduction}

In materials science, the stability of a structure is often dictated by an interaction potential which is attractive at long range and repulsive at short range ({\em e.g.} the Lennard-Jones potential).
For example, nanoparticles are nanometric crystals of spherical shape surrounded by a layer of molecules called ligands; they tend to aggregate under the action of attractive Van der Waals type forces and to repel each other due to the ligands (which can be thought of as sort of tiny springs).
Under certain circumstances, they self-assemble in an ordered way to form so-called supercrystals.
This is particularly interesting when the particles are of two (or more) different types and mix in a single phase.
It mays indeed allow to combine the physico-chemical properties of each type of particle in a single material.
This is what motivates here our interest in a possible phase separation phenomenon in such structures.

\begin{figure}[hbt]
\centering
\includegraphics[width=\textwidth]{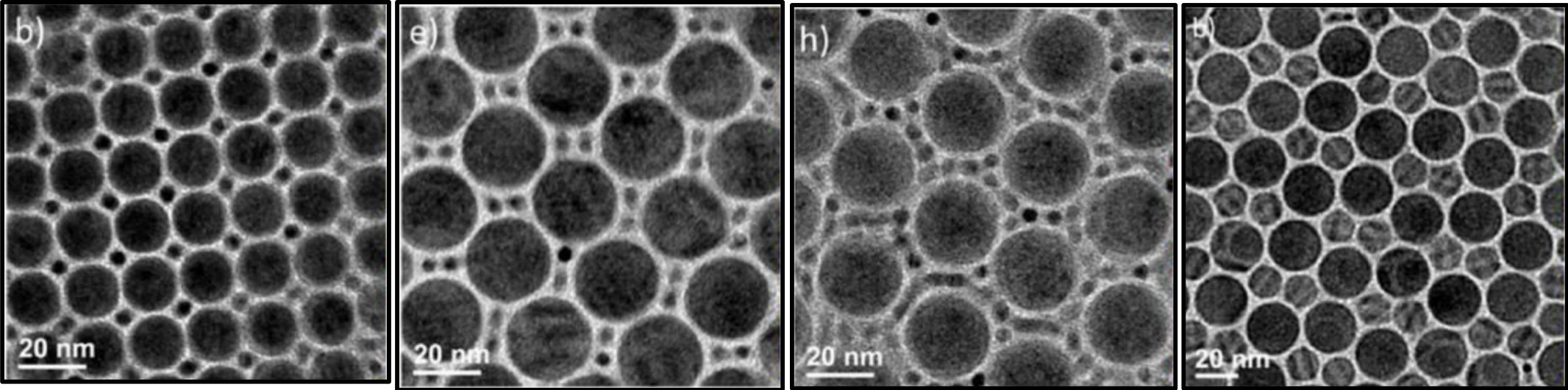}
\caption{
Supercrystals obtained in \cite{PDKM15} (courtesy of the authors).}
\label{fig:PDKM}
\end{figure}

A nice example of an ordered structure which harmoniously mixes different particles is provided by the binary supercrystals obtained in \cite{PDKM15} and reproduced in Figure~\ref{fig:PDKM}.
These supercrystals have moreover a remarkable property: they faithfully correspond to $4$ of the nine binary disc packings which have been proved in \cite{BF20} to maximize the density, compare with Figure~\ref{fig:triangulated}.
This suggests that minimizing the interaction energy amounts to maximizing the density.
Some mathematical results support this hypothesis (see, {\em e.g.}, \cite{Bet21,The06}), but the question still remains largely open.
This is what motivates here our interest in disc packings of maximal density, also called infinite-pressure hard-disc mixtures in the literature.

\begin{figure}[hbt]
\centering
\includegraphics[width=\textwidth]{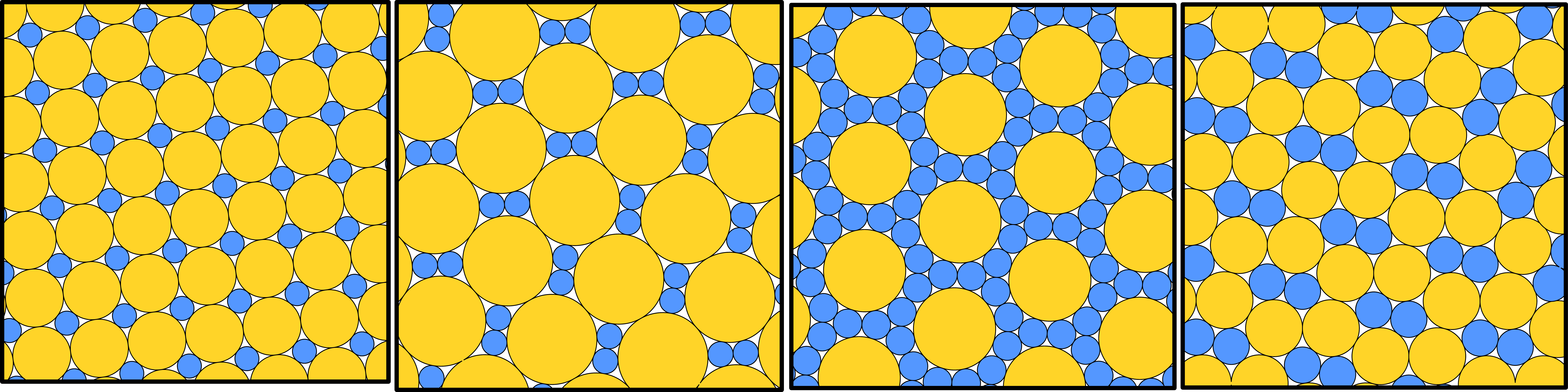}
\caption{Disc packings proven in \cite{BF20} to maximize the density.}
\label{fig:triangulated}
\end{figure}

In this article, we adopt a mathematical point of view.
We model structures by disc packings and restrict ourselves to the case of discs whose size ratio is $\sqrt{2}-1$, that is the first case on Fig.~\ref{fig:triangulated}.
Section~\ref{sec:results} is dedicated to the mathematical formalization of the problem and to the statement of the main results, namely Theorems~\ref{th:density} and \ref{th:packings}.

Theorem~\ref{th:density} explicitly gives the maximal density of a disc packing with disc sizes within ratio $\sqrt{2}-1$ as a function of the proportion of the large discs.
The claimed density is first proven to be a lower bound in Section~\ref{sec:lower} by using elementary combinatorics on words to construct packings which achieve this density.
The theorem follows once it is then proven to be also an upper bound.
This is done in Section~\ref{sec:upper}.
This is the essential part of this article, as upper bound on density of packings are usually hard to prove (think, {\em e.g.}, about the Hales-Ferguson proof of the Kepler conjecture \cite{Hal05}).
Moreover, there is a continuum of cases to be dealt with here (for each possible stoichiometry).
In a nutshell, we will show that any packing within each small enough interval of possible stoichiometry can be partitioned into cells whose densities - which can vary according to the cell - can be distributed between neighboring cells in such a way the average density over these neighbor cells is bounded from above by the claimed maximal density.
The proof is computer-assisted but nevertheless rigorous: it does not rely on numerical approximation but, instead, use intervals with exactly representable endpoints to deal with real numbers (see, {\em e.g.} \cite{Tuc11} for an introduction on rigorous computations).
It is adapted from the proof in \cite{BF20} with only a small improvement to handle the fixed stoichiometry constraint.
As the proof of \cite{BF20} is long and technical, we will not reproduce it here.
Instead, we will simply recall the main lines and insist on the modifications made.
This paper is thus {\em not} self-contained.

Theorem~\ref{th:packings} characterizes the local configurations which can appear with positive density in a densest packing.
This allows to describe how these densest packings look like.
It is proven in Section~\ref{sec:upper} and results from a further improvement of the proof in \cite{BF20}.

Last, in Section~\ref{sec:entropy}, we show that when there is no phase separation, {\em i.e.}, when there are more large discs, the densest packings are quite complicated in the sense that the number of different patterns they can form grows exponentially fast with the size of the considered patterns (Proposition~\ref{prop:entropy}).
The proof relies on a simple ``tiling widget''.

In this paper, we consider only the first case depicted in Fig.~\ref{fig:PDKM} or \ref{fig:triangulated}.
We conjecture that the second case is very similar, with a phase separation only for an excess of small discs, while the last two cases are less interesting, with a phase separation as soon as the stoichiometry differs from the one of the packings shown in Fig.~\ref{fig:triangulated}.
More generally, the ultimate goal would be to determine the whole phase diagram of binary disc packing, {\em i.e.}, to characterize the densest packings for any radius and any stoichiometry.
Strict bounds for any radius are given in \cite{Fer21}, without constraint on the stoichiometry.
A quite interesting though non-rigorous general picture of the phase diagram was proposed in \cite{LH93} and later improved in \cite{FJFS20}.

%%%%%%%%%%%%%%%%%%%%%
\section{Statement of the results}
\label{sec:results}

A {\em disc packing} (or {\em circle packing}) is a set of interior-disjoint discs in the Euclidean plane.
Its {\em density} $\delta$ is the proportion of the plane covered by the discs:
$$
\delta:=\limsup_{k\to \infty}\frac{\textrm{area of the square $[-k,k]^2$ covered by discs}}{\textrm{area of the square $[-k,k]^2$}}.
$$
We consider packings by discs of radius $1$ and $r:=\sqrt{2}-1$, simply called ``packings'' afterwards.
Discs of radius $1$ are called {\em large discs} while discs of radius $r$ are called {\em small discs}.
The packings with a proportion $x$ of large discs are called {\em $x$-packings}.
We denote by $\delta(x)$ the supremum of the densities over all the $x$-packings.
On the one hand, it was proven in \cite{FT43}
$$
\delta(0)=\delta(1)=\frac{\pi}{2\sqrt{3}}\approx 0.9069.
$$
On the other hand, it was proven in \cite{Hep00} (see also \cite{Hep03,BF20})
$$
\forall x\in[0,1],\quad\delta(x)\leq\delta(\tfrac{1}{2})=\frac{\pi}{2+\sqrt{2}}\approx 0.9202.
$$
These densities are reached by the periodic packings depicted in Fig.~\ref{fig:packings_known}.
Here, we extend this to any stoichiometry $x$.

\begin{figure}[hbt]
\centering
\includegraphics[width=\textwidth]{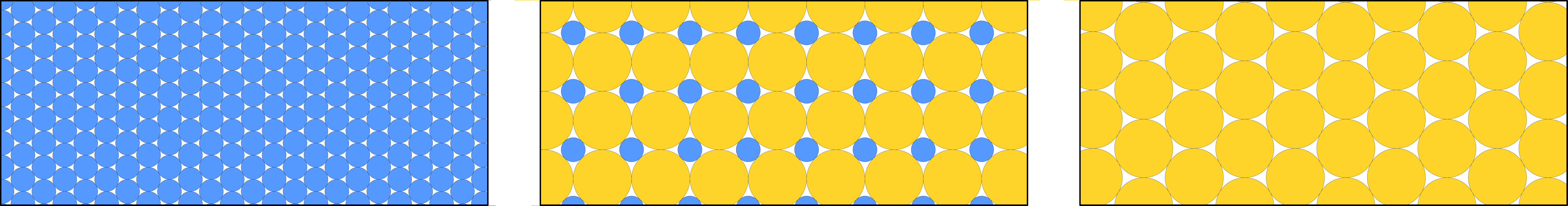}
\caption{Periodic packings of maximal density for a proportion $0$, $\tfrac{1}{2}$ and $1$ of large discs.}
\label{fig:packings_known}
\end{figure}

\begin{theorem}
\label{th:density}
For $0\leq x\leq 1$, the maximal density $\delta(x)$ of $x$-packings is (Fig.~\ref{fig:density}):
$$
\delta(x\leq\tfrac{1}{2})=\frac{\pi(x+(1-x)r^2)}{4x+2(1-2x)r^2\sqrt{3}},
\qquad
\delta(x\geq\tfrac{1}{2})=\frac{\pi(x+(1-x)r^2)}{4(1-x)+2(2x-1)\sqrt{3}}.
$$
%Moreover, for any $x$, there exists a packing with density $\delta(x)$.
\end{theorem}

\begin{figure}[hbt]
\centering
\includegraphics[width=\textwidth]{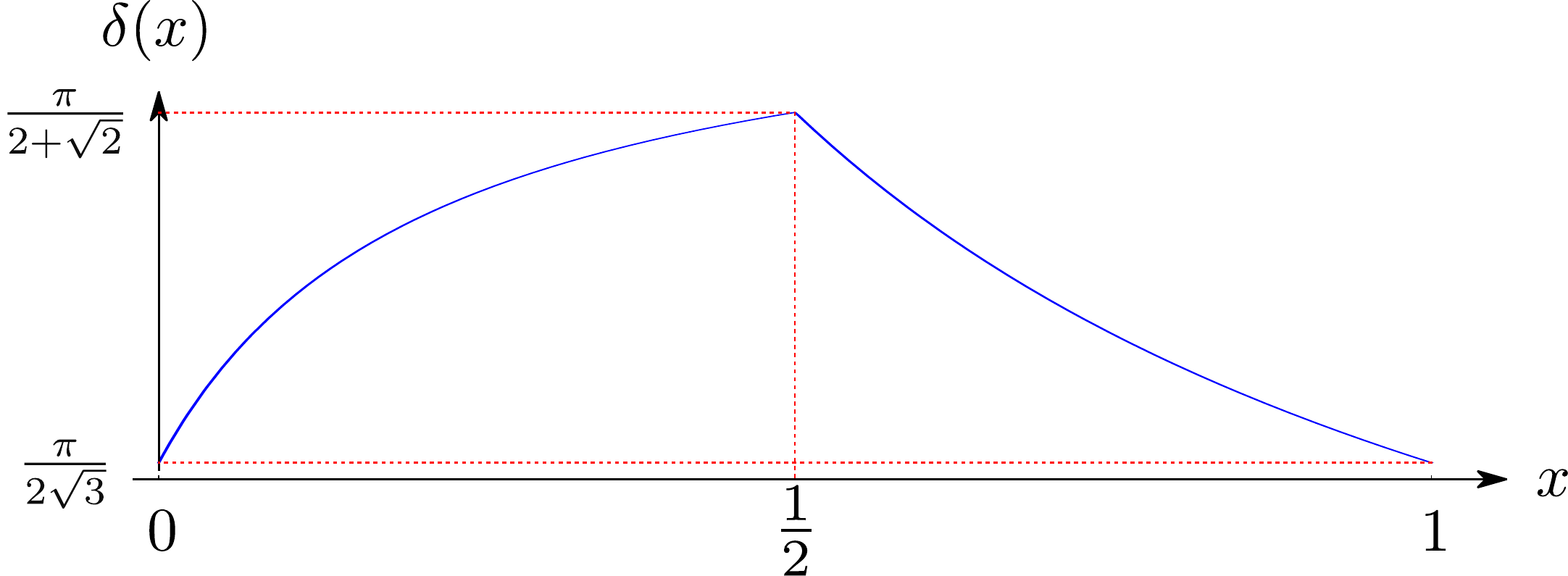}
\caption{The maximal density as a function of the proportion $x$ of large discs.}
\label{fig:density}
\end{figure}

Further, we would like to describe the set of all the densest $x$-packings.
One difficulty that arises is that, whatever packing we consider, there are continuously many packings with the same density that are different, even though they look pretty much the same (Fig.~\ref{fig:packings_known2}).
We need to formalize this similiarity.

\begin{figure}[hbt]
\centering
\includegraphics[width=\textwidth]{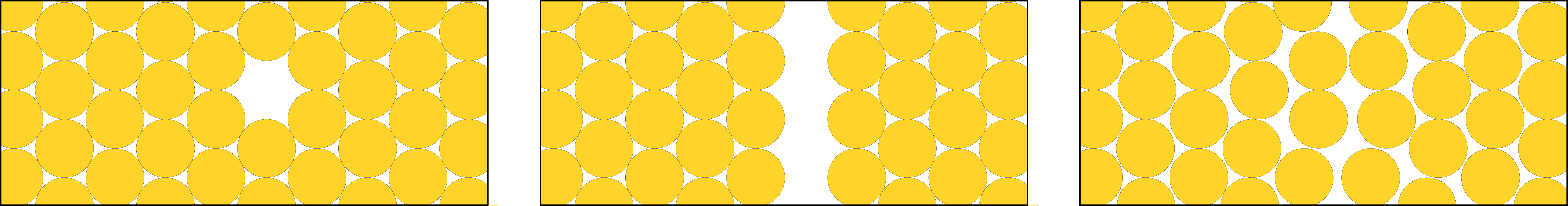}
\caption{
Removing a disc in a packing does not affect the density (left).
This holds for infinitely many discs, as long as they are in negligible proportion (center).
Remaining discs can then be slightly moved (even arbitrarily far ones) without affecting the density (right).
}
\label{fig:packings_known2}
\end{figure}

Given a disc packing, define the {\em cell} of a disc as the set points of the plane which are closer to this disc than to any other (the boundary between two cells is thus a branch of hyperbola).
These cells form a partition of the plane whose dual is a triangulation\footnote{In the non-generic case where there is a point in the plane equidistant from $k>3$ discs, and thus a $k>3$-sided face in the dual, this face is triangulated in an arbitrary way.}, referred to as the {\em FM-triangulation} of the packing.
Introduced in \cite{FM58} (see also \cite{FT64}), FM-triangulations are also known as {\em additively weighted Delaunay triangulations}.
For example, the FM-triangulations of the leftmost and rightmost packings in Fig.~\ref{fig:packings_known} are triangular grids (the one of the central packing is a so-called {\em tetrakis square tiling}).
A more generic example is depicted in Fig.~\ref{fig:FM_triangulation}.

\begin{figure}[hbt]
\centering
\includegraphics[width=0.47\textwidth]{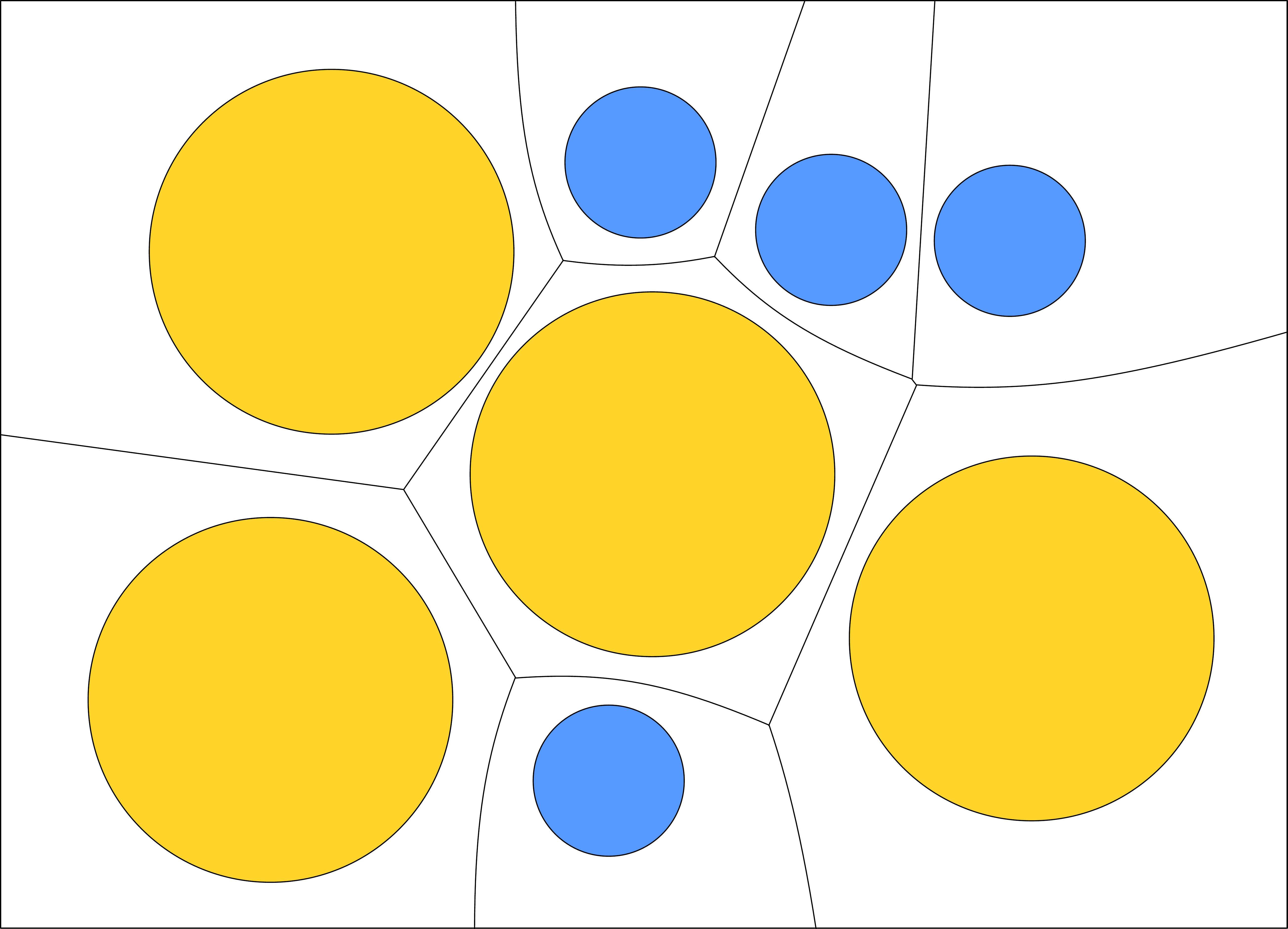}
\hfill
\includegraphics[width=0.47\textwidth]{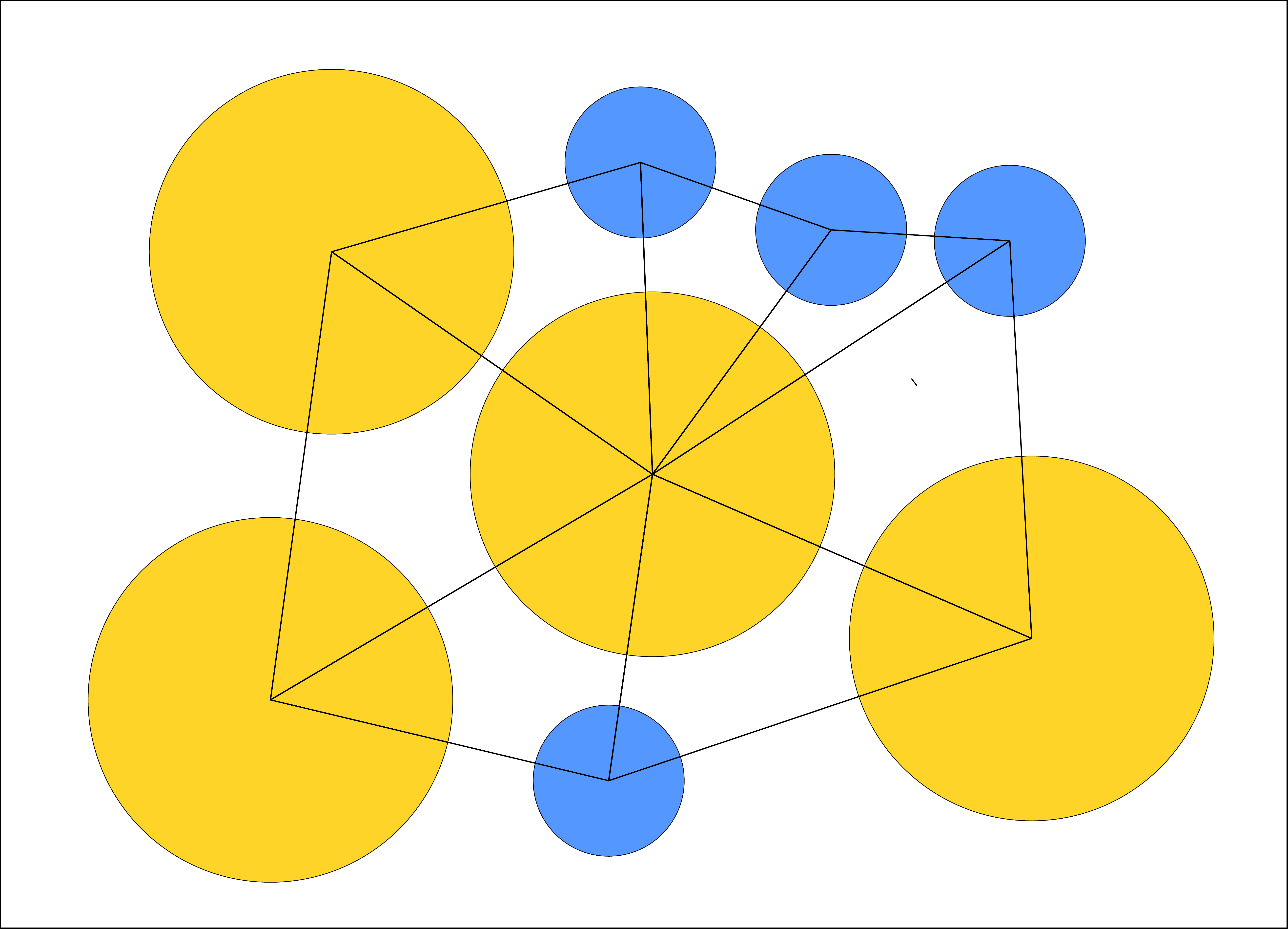}
\caption{
Some discs in the plane and their cells (left).
The corresponding dual graph defines the FM-triangulation (right).
}
\label{fig:FM_triangulation}
\end{figure}

Now, let us call {\em neighborhood} of a disc its neighbors in the FM-triangulation, ordered clockwise and up to a circular permutation.
Such a neighborhood is coded by the word over the alphabet $\{1,r\}$ which gives the radii of the neighbor discs.
For example, in Fig.~\ref{fig:packings_known}, each small disc has the neighborhood rrrrrr in the leftmost packing or 1111 in the central packing, while each large discs has the neighborhood 1r1r1r1r in the central packing or 111111 in the rightmost packing.
This allows to state out second result:

\begin{theorem}
\label{th:packings}
Consider an $x$-packing with density $\delta(x)$.
If $x\leq\tfrac{1}{2}$, then almost any small (resp. large) disc has a neighborhood 1111 or rrrrrr (resp. 1r1r1r1r).
If $x\geq \tfrac{1}{2}$, then almost any small (resp. large) disc has a neighborhood 1111 (resp. 1r1r1r1r, 1111r1r, 111r11r or 111111).
\end{theorem}

For $x\leq \tfrac{1}{2}$, this result (which may seem confusing at first glance) yields a good insight into the look of the densest packings.
Consider indeed a small disc.
If its neighborhood is rrrrrr, then each of these $6$ small neighbor discs has itself three small neighbor discs, hence a  neighborhood rrrrrr (since the only other neighborhood, 1111, is no more possible).
Continuing from neighbor to neighbor unveils the hexagonal compact packings depicted in Fig.~\ref{fig:packings_known}, left.
If, on the contrary, its neighborhood is 1111, then each of the $4$ small neighbor discs has itself neighborhood 1111, and each of the $4$ large neighbor discs has neighborhood 1r1r1r1r.
Continuing from neighbor to neighbor this times unveils discs packed as in Fig.~\ref{fig:packings_known}, center.
Both cases can actually coexist in the same packing, because Theorem~\ref{th:packings} deals only with {\em almost} any disc: the negligible proportion of discs with other neighborhood can play the role of a ``joint'' between large regions of near-periodically packed discs.
The packing then looks like what is called in materials science a {\em twinned crystal}, {\em i.e.}, an aggregate of different crystalline domains (Fig.~\ref{fig:twinning}).

\begin{figure}[hbt]
\centering
\includegraphics[width=\textwidth]{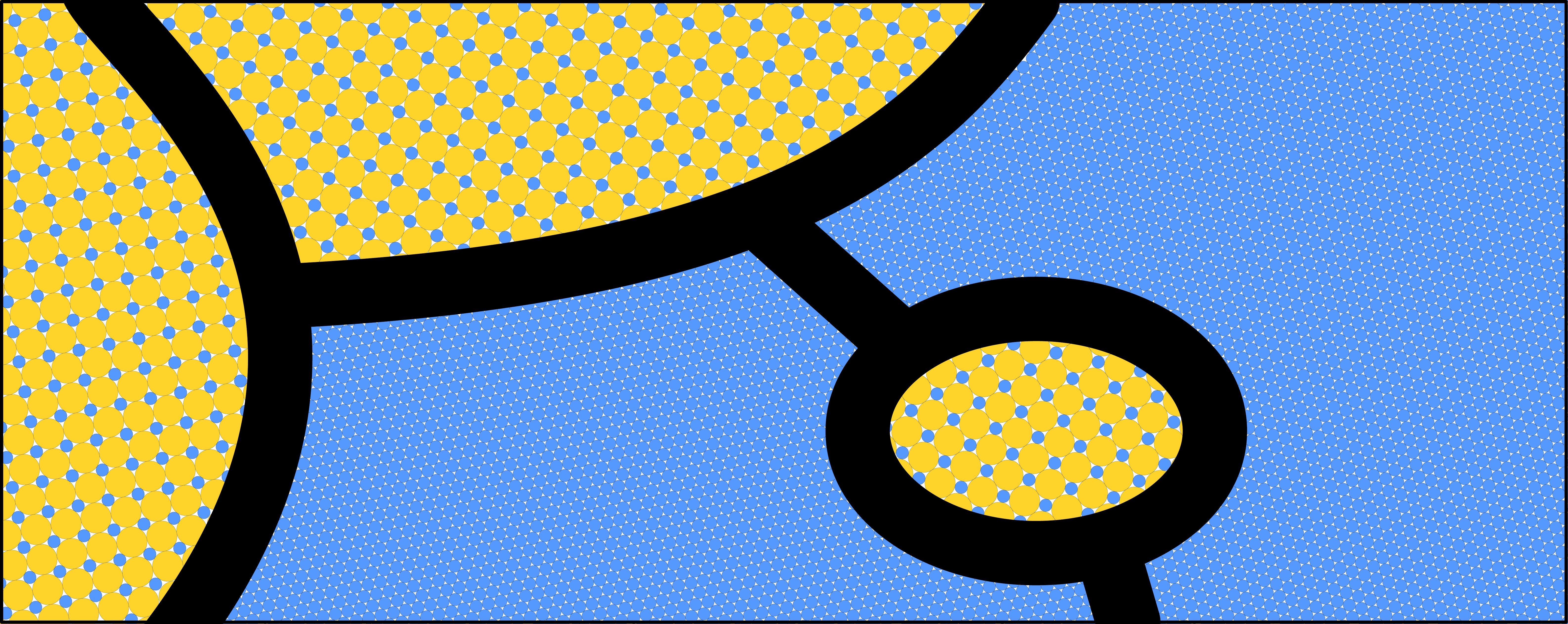}
\caption{
Typical look of a densest packing with an excess of small discs, with thick black lines covering the negligible proportion of discs with a neighborhood not listed in Th.~\ref{th:packings}.
}
\label{fig:twinning}
\end{figure}

For $x>\tfrac{1}{2}$, we also get aggregates of large regions in which only the neighborhoods mentioned in Theorem~\ref{th:packings} appear.
But these regions may not be near-periodically packed anymore.
They actually look like so-called {\em square-triangle tilings}.
These are coverings of the plane by interior disjoint squares and regular triangles allowed to intersect only in a single point or a on a whole edge (Fig.~\ref{fig:tiles2discs}, left).
Any square-triangle tiling can indeed be transformed into a disc packing (Fig.~\ref{fig:tiles2discs}, right).
A simple computation shows that if there is a proportion $2x-1$ of squares, then we get an $x$-packing of density $\delta(x)$.
Conversely, the possible disc neighborhoods in an $x$-packing of density $\delta(x)$ for $x\geq\tfrac{1}{2}$ ensure it can be transformed into a square-triangle tiling (Fig.~\ref{fig:discs2tiles}).
However, one could argue about how much this correspondence really sheds light on the set of densest packings.
Indeed, the set of square-triangle tilings with a given proportion of squares is itself quite complicated, as it will be illustrated by Proposition~\ref{prop:entropy}.

\begin{figure}[hbt]
\centering
\includegraphics[width=\textwidth]{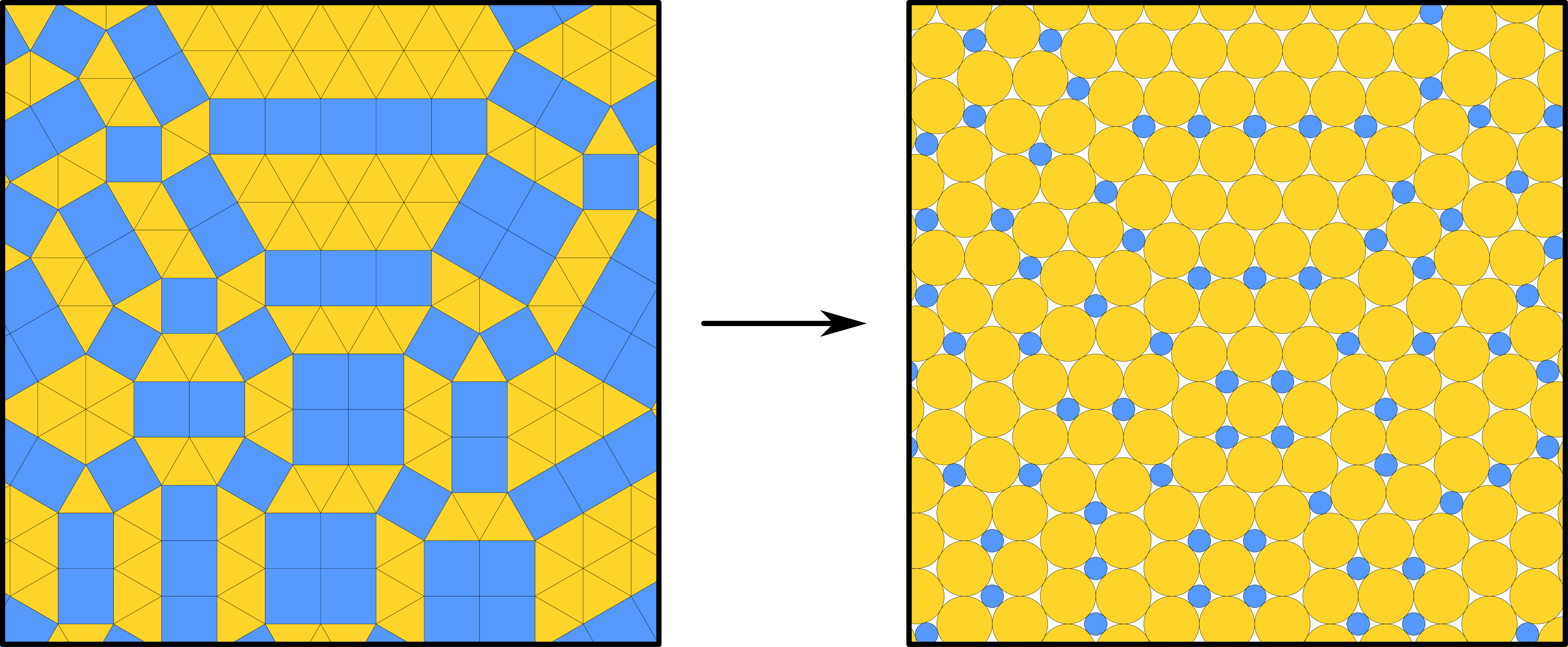}
\caption{
A square-triangle tiling (left) can be transformed into a packing by discs of radius $1$ and $r=\sqrt{2}-1$ (right) by putting a large disc on each vertex and a small disc at the center of each square.
}
\label{fig:tiles2discs}
\end{figure}

\begin{figure}[hbt]
\centering
\includegraphics[width=\textwidth]{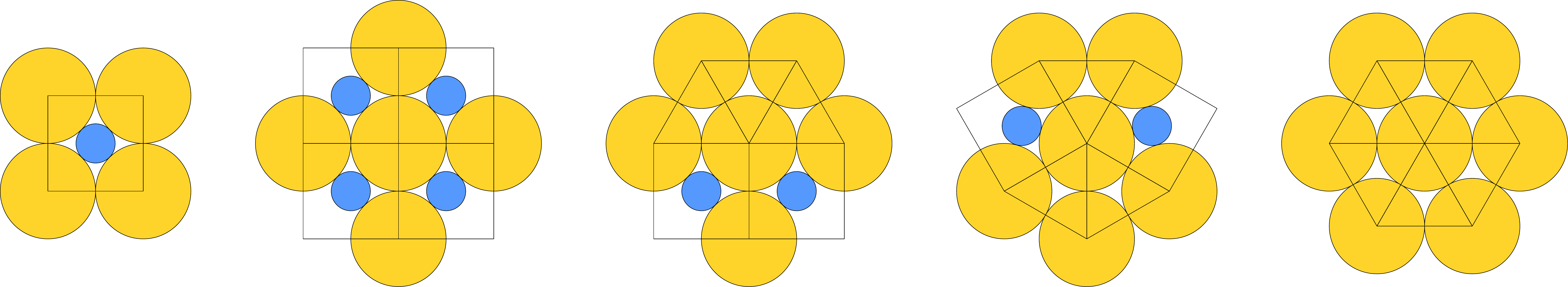}
\caption{
A disc packing with only suitable disc neighborhoods (from left to right: 1111, 1r1r1r1r, 1111r1r, 111r11r and 111111) can be transformed back into a square-triangle tiling.
}
\label{fig:discs2tiles}
\end{figure}

\begin{proposition}
\label{prop:entropy}
Call {\em pattern of size $k$} of a square-triangle tiling its restriction to its tiles which lies into some ball of radius $k$.
Then, for any given $\alpha\in(0,1)$, the number of different patterns which can appear in a square-triangle tiling with $\alpha$ squares for $1-\alpha$ triangles grows exponentially in $k^2$.
\end{proposition}

In other words, the set of square-triangle tilings with fixed proportions of tiles has {\em positive entropy}.
Determining rigorously the value of the base of the growth exponent is an open problem - it is expected to be maximal for the tilings with a ratio of $\sqrt{3}$ squares for $4$ large triangles, known as {\em $12$-fold quasicrystals} \cite{ICJKS21,Wid93,Nie98,Kal99}.

%%%%%%%%%%%%%%%%%%%%%
%%%%%%%%%%%%%%%%%%%%%
%%%%%%%%%%%%%%%%%%%%%
\section{Lower bound}
\label{sec:lower}

We here explicitly build an $x$-packing with density $\delta(x)$.
We rely on elementary combinatoric on words.
For $\alpha\in [0,1)$, denote by $u(\alpha)$ the sequence of $\{0,1\}^{\mathbb{Z}}$ whose $k$-th letter $u_k$ is defined by
$$
u_k=0
\quad\Leftrightarrow\quad
k\alpha \mod 1\in [0,1-\alpha).
$$
For example (the bold letter has index $0$):
\begin{eqnarray*}
u(\tfrac{1}{3})&=&\cdots 01001001001001001001{\bf 0}0100100100100100100 \cdots\\
u(\sqrt{2}-1)&=&\cdots   10010101001010010101{\bf 0}0101001010010101001 \cdots
\end{eqnarray*}
The case $\alpha\notin\mathbb{Q}$ corresponds to a so-called {\em Sturmian word} introduced in \cite{MH40}.
Here, we will only use the fact that $u(\alpha)$ has a proportion $\alpha$ of letter $0$.

For $\tfrac{1}{2}< x\leq 1$, we use Sturmian words to mix the central and rightmost packings depicted in Fig.~\ref{fig:packings_known}.
Let $\alpha:=\tfrac{1-x}{x}\in [0,1)$.
We associate with $u(\alpha)$ the square-triangle tiling made of vertical columns of either squares of triangles, with the $k$-th column being made of squares if and only if the $k$-th letter of $u(\alpha)$ is $0$ (Fig.~\ref{fig:optimal_packing}).
This tiling has $\alpha$ squares for $2(1-\alpha)$ triangles.
Putting a large disc on each vertex and a small disc in the center of each square (recall Fig.~\ref{fig:tiles2discs}) yields a packing with $1$ large disc for $\alpha$ small discs, {\em i.e.}, an $x$-packing.
Its density is
$$
\frac{\alpha(\pi+\pi r^2)+2(1-\alpha)\tfrac{\pi}{2}}{4\alpha+2(1-\alpha)\sqrt{3}}=\frac{\pi(x+(1-x)r^2)}{4(1-x)+2(2x-1)\sqrt{3}}=\delta(x).
$$

\begin{figure}[hbt]
\centering
\includegraphics[width=\textwidth]{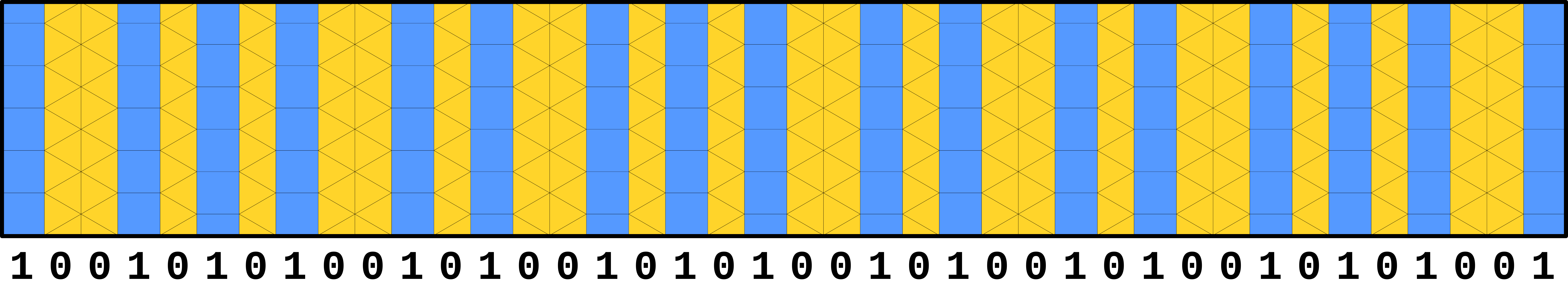}
\caption{The square-triangle tiling associated with the sequence $u(\sqrt{2}-1)$.}
\label{fig:optimal_packing}
\end{figure}

For $0<x<\tfrac{1}{2}$, we want to use alike Sturmian words to mix the leftmost and central packings depicted in Fig.~\ref{fig:packings_known}.
The situation is however more complicated because these packings do not ``mix well'' anymore: we have to combine two types of regions and deal with the ``joint'' (recall Fig.~\ref{fig:twinning}).
For this purpose, we introduce the following transformation on sequences.
If $u=(u_k)_{k\in\mathbb{Z}}$, then $\widehat{u}$ denotes the sequence obtained by replacing the $k$-th letter of $u$ by the same letter repeated $|k|+1$ times.
For example (the bold letter has index $0$):
\begin{eqnarray*}
\widehat{u}(\tfrac{1}{3})&=&\cdots 00000011111000000011{\bf 0}0011100000000011111\cdots\\
\widehat{u}(\sqrt{2}-1)&=&\cdots   00000111110000111001{\bf 0}0011100001111100000\cdots
\end{eqnarray*}

\begin{lemma}
\label{lem:word_expansion}
The transformation $u\mapsto \widehat{u}$ does not modify letter proportions.
\end{lemma}

\begin{proof}
Consider a factor $w$ of $\widehat{u}$ and a letter $a$ which has frequency $\alpha$ in $u$.
We assume it has only positive indices (the case of negative indices is symmetric, and if it has both positive and negative indices we break it into two parts).
It can be written
$$
w=pu_i^{i+1}u_{i+2}^{i+2}\cdots u_{i+k}^{i+k+1}s,
$$
where the prefix $p$ (resp. the suffix $s$) has length at most $i-1$ (resp. $i+k+1$).
By reordering the letters, $w$ is rewritten as
$$
p(u_i\cdots u_{i+k})^{i+1}(u_{i+1}\cdots u_{i+k})(u_{i+2}\cdots u_{i+k})\cdots u_{i+k}s.
$$
Fix $i$ and let $k$ grow.
Since $w$ has length of order $k^2$, the prefix $p$ and the suffix $s$ do not affect the proportion of $a$ in $w$.
By hypothesis, the proportion of $a$ in $u_i\cdots u_{i+k}$ - hence in $(u_i\cdots u_{i+k})^{i+1}$ - tends towards $\alpha$.
The Stolz-Cesàro theorem ensures that the proportion of $a$ in $(u_{i+1}\cdots u_{i+k})(u_{i+2}\cdots u_{i+k})\cdots u_{i+k}$ tends towards $\alpha$.
The proportion of $a$ in $w$ thus also tends towards $\alpha$
\end{proof}

We then associate with any sequence $u\in\{0,1\}^\mathbb{Z}$ a disc packing as follows (Fig.~\ref{fig:suboptimal_packing}).
It is made of vertical column of identical adjacent discs.
The columns alternate horizontally as the letters in $u$, with a colum of large discs for a letter $0$ or a column of small discs for a letter $1$.
Consecutive columns of small discs are disposed so that small discs form an hexagonal compact packing as in the leftmost packing depicted in Fig.~\ref{fig:packings_known}.
Consecutive columns of large discs form a square grid and a small disc is inserted in each hole between four large discs, as in the central packing depicted in Fig.~\ref{fig:packings_known}.

\begin{figure}[hbt]
\centering
\includegraphics[width=\textwidth]{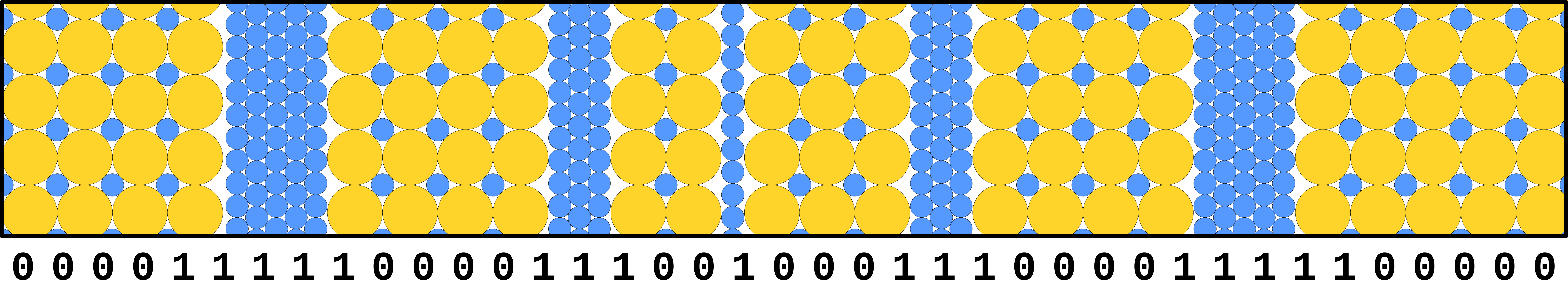}
\caption{The packing associated with the sequence $\widehat{u}(\sqrt{2}-1)$.}
\label{fig:suboptimal_packing}
\end{figure}

Let $\alpha:=\tfrac{x}{1-x}\in [0,1)$.
The sequence $u(\alpha)$ has a proportion $\alpha$ of $0$.
The packing associated with $u(\alpha)$ has $\alpha$ large discs for $1$ small disc, {\em i.e.}, it is an $x$-packing.
However, each factor $01$ or $10$ in a sequence $u$ corresponds in the associated packing to a ``joint'' between two regions which locally decreases the density.
In $u(\alpha)$, there is way too much $01$ and $10$ (two for each small disc) in order to reach the density $\delta(x)$.
This is where the above transformation comes in.
Lemma~\ref{lem:word_expansion} ensures that $\widehat{u}(\alpha)$ yields an $x$-packing as well as $u(\alpha)$, but it has moreover a negligible proportion of factors $01$ and $10$ (a factor of length $k$ contains $O(\sqrt{k})$ factors $01$ and $10$).
The two types of region thus respectively cover a proportion $\alpha$ and $1-\alpha$ of the plane.
The density is
$$
\frac{\alpha\pi+ \pi r^2}{4\alpha+2(1-\alpha)r^2\sqrt{3}}=\frac{\pi(x+(1-x)r^2)}{4x+2(1-2x)r^2\sqrt{3}}=\delta(x).
$$
This proves that the density given in Theorem~\ref{th:density} is a lower bound on the maximal density.

%%%%%%%%%%%%%%%%%%%%%
%%%%%%%%%%%%%%%%%%%%%
%%%%%%%%%%%%%%%%%%%%%
\section{Upper bound and densest packings}
\label{sec:upper}

As mentioned in the introduction, we follows the strategy of \cite{BF20} with only a small improvement to handle the fixed stoichiometry constraint.
The strategy used in \cite{BF20} resembles, though less complicated, the one used by Hales to prove the Kepler conjecture \cite{Hal05}.
We sketch it in Subsection~\ref{sec:strategy} with just enough detail to explain the improvement made here.
Then, in Subsection~\ref{sec:stoichiometry}, we sketch the proof of the upper bound claimed in Theorem~\ref{th:density}.
The interested reader is referred to Appendix~\ref{sec:details}, where are given all the technical details that allow to verify the proof {\em in combination with \cite{BF20}}.
Last, Subsection~\ref{sec:densest_packings} proves Theorem~\ref{th:packings}.

%%%%%%%%%%%%%%%%%%%%%
\subsection{Strategy of \cite{BF20}}
\label{sec:strategy}

Let $\delta$ be the candidate upper bound on the maximal density.
Fix a packing and consider an FM-triangulation of its disc centers.
Because density is not additive over triangles, it is convenient to introduce the {\em emptiness} $E$, defined for any triangle $T$ of the FM-triangulation by
$$
E(T):=\delta\times\textrm{area}(T)-\textrm{cov}(T),
$$
where $\textrm{area}(T)$ is the area of $T$ and $\textrm{cov}(T)$ is the area of $T$ covered by the discs centered on the vertices of $T$.
In particular, the density within a triangle $T$ (that is, the proportion of its area covered by discs) is less than or equal to $\delta$ if and only if its emptiness is non-negative.
To prove that the density of the considered packing is bounded from above by $\delta$, it thus suffices to prove that the emptiness, averaged over all the triangles, is non-negative.

But how to prove an inequality that involves infinitely many triangles?
The principle, called {\em localization} in \cite{Lag02}, consists in distributing the emptiness of each triangle in an intelligent way among its three vertices so that the total emptiness received by each vertex of the triangulation is non-negative.
This suffices to ensure that the emptiness averaged over any triangulation is non-negative.

Sadly, there are also infinitely many possible configurations around a vertex since the discs have coordinates in $\mathbb{R}^2$.
However, the properties of FM-triangulations ensure that the set of configurations to consider is {\em compact}.
This allows to perform rigorous calculations using interval arithmetic.
Many details are here omitted, but this gives a fairly good idea of the strategy of \cite{BF20}.

%%%%%%%%%%%%%%%%%%%%%
\subsection{Stoichiometry comes into play}
\label{sec:stoichiometry}

Let us now turn to packings with a fixed proportion $x$ of large discs ($x$-packings).
Instead of proving that the emptiness received by each vertex of the triangulation is at least $0$, we will prove that there exist real numbers $\alpha_1$ and $\alpha_r$ which satisfy
\begin{equation}
\label{eq:proportions}
x\alpha_1+(1-x)\alpha_r= 0
\end{equation}
and such that the emptiness received by the center of a disc of radius $q\in\{1,r\}$ is at least $\alpha_q$.
This will ensure that the emptiness averaged over any triangulation that has $x$ large discs for $1-x$ small discs is non-negative.
Note that $\alpha_1$ and $\alpha_r$ are necessarily of opposite signs: the local density around a disc depends on the radius of this disc!

Further, this has to be done for any $x\in[0,1]$.
Again, we rely on the fact that $[0,1]$ is compact to perform rigorous computations.
Namely, we cut $[0,1]$ in intervals which are sufficiently smalls so that the inequality around each vertex can be rigorously ensured by interval arithmetic computations (as few as $100$ intervals of length $0.01$ appeared to be sufficient).

This proves that the density given in Theorem~\ref{th:density} is an upper bound on the maximal density (see Appendix~\ref{sec:details} for details).

%%%%%%%%%%%%%%%%%%%%%
\subsection{Densest packings}
\label{sec:densest_packings}

For lighter wording, let us call {\em bad neighborhood} the following neighborhoods:
\begin{itemize}
\item if $x\leq\tfrac{1}{2}$, the neighborhood of a small (resp. large) disc other than 1111 or rrrrrr (resp. 1r1r1r1r);
\item if $x\geq \tfrac{1}{2}$, the neighborhood of a small (resp. large) disc other than 1111 (resp. 1r1r1r1r, 1111r1r, 111r11r or 111111).
\end{itemize}
Proving Theorem~\ref{th:packings} thus amounts to prove that in any $x$-packing of density $\delta(x)$, the proportion of discs with a bad neighborhood is zero.

Once again, a very slight modification in the computer-assisted proof is all that is required.
Instead of only proving that the emptiness received by each vertex of the triangulation is at least $\alpha_q$ for the center of any disc of size $q\in\{1,r\}$, we also prove that there exists $\eta>0$ such that the emptiness received by any center of a disc of size $q$ with a bad neighborhood is at least $\alpha_q+\eta$.
Each bad neighborhood thus causes a local loss of density, which is therefore only possible for a zero proportion of discs in a maximum density packing.
More precisely, this ensures that in any $x$-packing with density $\delta\leq\delta(x)$, the proportion of discs that have a bad neighborhood is at most
$$
\frac{\delta(x)-\delta}{\eta}.
$$
In other words, the amount of defects in an packing is linearly bounded by the gap between its density and the maximal density.

To find a suitable $\eta$, we just add a candidate value for $\eta$ to both $\alpha_1$ and $\alpha_r$ in the program and rerun the computation to check whether all the inequalities still hold.
For a partition of $[0,1]$ in $100$ intervals, we found that the value $\eta=10^{-4}$ works (smaller intervals allow to improve this value a little bit, at the cost of a longer calculation time).
This proves Theorem~\ref{th:packings}.

%%%%%%%%%%%%%%%%%%%%%
%%%%%%%%%%%%%%%%%%%%%
%%%%%%%%%%%%%%%%%%%%%
\section{Positive entropy}
\label{sec:entropy}

We here prove Proposition~\ref{prop:entropy}, which shows that the densest packings with a given excess of large discs form a rather complicated set.\\

\begin{proof}
Let $\alpha\in(0,1)$. % $\alpha=2x-1\in (0,1)$
There are two different ways to tile a regular dodecagon by squares and regular triangles (Fig.~\ref{fig:entropy1}).
Any square-triangle tiling with a positive frequency of dodecagons will thus yields a number of different patterns of size $k$ which grows exponentially in $k^2$ (by tiling independently each dodecagon).
We shall define such a tiling with $\alpha$ squares for $1-\alpha$ large triangles.
%, so that applying the transformation depicted in Fig.~\ref{fig:tiles2discs} will yield a packing with $x$ large discs for $1-x$ small discs.

The idea is to replace, in a square-triangle tiling associated with a sequence $u(\beta)$ as explained in Section~\ref{sec:lower}, each square and triangle by respectively, the building blocks $\mathcal{S}_n$ and $\mathcal{T}_n$ depicted in Fig.~\ref{fig:entropy2}, and each vertex by a dodecagon depicted in Fig.~\ref{fig:entropy1}.
The point is to find suitable $\beta$ and $n$.
The numbers $s_n^\square$ of squares and $s_n^\triangle$ of triangle in $\mathcal{S}_n$ (resp. $t_n^\square$ and $t_n^\triangle$ in $\mathcal{T}_n$) are
\begin{eqnarray*}
s_n^\square&=&(n+1)^2+\tfrac{4n}{2}+6=n^2+4n+7\\
s_n^\triangle&=&4(2n+1)+12=8n+16,\\
t_n^\square&=&\tfrac{3n}{2}+\tfrac{6}{2}=\tfrac{3}{2}n+3,\\
t_n^\triangle&=&3+5+\ldots+(2n+1)+\tfrac{12}{2}=n^2+2n+6.
\end{eqnarray*}
The ratio of squares and triangles in the resulting tiling is
$$
f(\beta,n):=\frac{\beta s_n^\square+(1-\beta)t_n^\square}{\beta s_n^\triangle+(1-\beta)t_n^\triangle}.
$$
In particular
\begin{eqnarray*}
\lim_{n\to\infty}f(0,n)&=&\lim_{n\to\infty}\frac{\tfrac{3}{2}n+3}{n^2+2n+6}=0,\\
\lim_{n\to\infty}f(1,n)&=&\lim_{n\to\infty}\frac{n^2+4n+7}{8n+16}=+\infty.
\end{eqnarray*}
Since $\beta\mapsto f(\beta,n)$ is continuous, this ensures that for any $\alpha\in(0,1)$, {\em i.e.}, for any $\tfrac{\alpha}{1-\alpha}\in(0,\infty)$, for $n$ large enough there exists $\beta\in(0,1)$ such that $f(\beta,n)=\tfrac{\alpha}{1-\alpha}$.
For such $\beta$ and $n$, replacing the squares and triangles of the tiling associated with $u(\beta)$ by the building blocks $\mathcal{S}_n$ and $\mathcal{T}_n$ yields a tiling with $\alpha$ squares for $1-\alpha$ triangles.%, hence the wanted $x$-packing.
\end{proof}

\begin{figure}[hbt]
\centering
\includegraphics[width=0.65\textwidth]{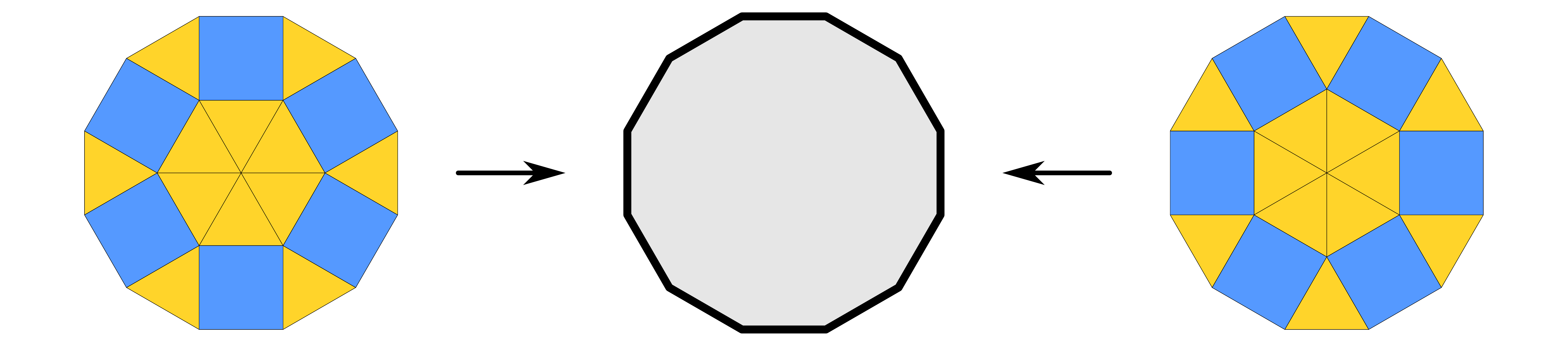}
\caption{The two ways to tile a regular dodecagon.}
\label{fig:entropy1}
\end{figure}

\begin{figure}[hbt]
\centering
\includegraphics[width=0.65\textwidth]{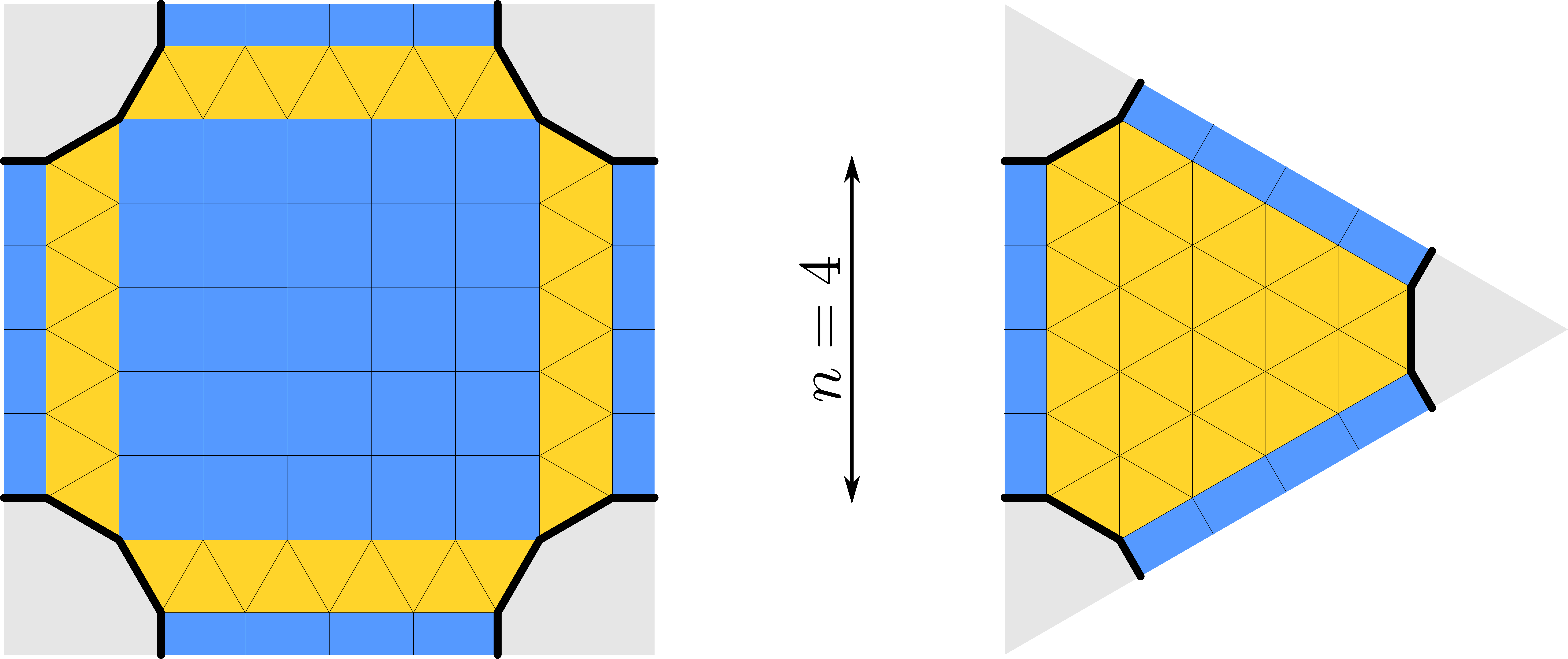}
\caption{
The building blocks $\mathcal{S}_n$ and $\mathcal{T}_n$ for $n=4$.
}
\label{fig:entropy2}
\end{figure}

\appendix
%%%%%%%%%%%%%%%%%%%%%%%%
%%%%%%%%%%%%%%%%%%%%%%%%
%%%%%%%%%%%%%%%%%%%%%%%%
\section{Proof of the upper bound on the density}
\label{sec:details}

We use here the notations and concepts of \cite{BF20}, of which this article can thus be seen as a companion article.
We apologize to the reader for the extra work involved in carefully checking the proof, but it seemed absurd to copy and paste pages from a previous article under the pretext of being self-contained.

Fix a proportion $x$ of large discs.
Fix an $x$-packing and an FM-triangulation $\mathcal{T}$ of the center of its discs.
Recall that the {\em emptiness} $E(T)$ of any triangle $T$ is defined by
$$
E(T):=\delta(x)\times\textrm{area}(T)-\textrm{cov}(T),
$$
where $\textrm{area}(T)$ is the area of $T$, $\textrm{cov}(T)$ is the area of $T$ inside the discs centered on the vertices and $\delta(x)$ is the density defined in Th.~\ref{th:density}.
To prove that our packing has density at most $\delta(x)$, we will prove
\begin{equation}
\label{eq:globalE}
\sum_{T\in \mathcal{T}}E(T)\geq 0.
\end{equation}
For this, we will define a potential $U$ on triangles satisfying the {\em global inequality}
\begin{equation}
\label{eq:global}
\sum_{T\in \mathcal{T}}U(T)\geq 0,
\end{equation}
and, for any triangle $T\in\mathcal{T}$, the {\em local inequality}
\begin{equation}
\label{eq:local}
E(T)\geq U(T).
\end{equation}
The term {\em Global inequality} is used here for an inequality that must be verified on the whole packing, as opposed to a {\em local inequality} that must be verified on each possible triangle, independently of each other.

The potential $U$ will be the sum of a {\em vertex potential} $U_v(T)$, defined for any vertex $v$ of $T$, and an {\em edge potential} $U_e(T)$, defined for any edge $e$ of $T$.
The edge potential shall satisfy the inequality
\begin{equation}
\label{eq:edge}
\sum_{T\in\mathcal{T}|e\in T} U_e(T)\geq 0.
\end{equation}
The vertex potential must satisfy a {\em modified version} of the {\em vertex inequality} of \cite{BF20}, namely
\begin{equation}
\label{eq:vertex}
\sum_{T\in\mathcal{T}|v\in T}U_v(T)\geq \alpha_q,
\end{equation}
where $q\in\{1,r\}$ is the radius of the disc of center $v$ and the two real numbers $\alpha_1$ and $\alpha_r$ satisfy
\begin{equation}
\label{eq:proportions}
x\alpha_1+(1-x)\alpha_r=0.
\end{equation}
Let us stress that this is {\em the one and only difference} with the strategy of \cite{BF20}!
The global inequality~\eqref{eq:global} then follows:
$$
\sum_{T\in \mathcal{T}}U(T)=\sum_{e\in \mathcal{T}}\underbrace{\sum_{T\ni e} U_e(T)}_{\geq 0}+\underbrace{\sum_{v\in \mathcal{T}\atop q=1}\underbrace{\sum_{T\ni v} U_v(T)}_{\geq\alpha_1}+\sum_{v\in \mathcal{T}\atop q=r}\underbrace{\sum_{T\ni v} U_v(T)}_{\geq \alpha_r}}_{\textrm{$\geq 0$ by disc proportions and Eq.~\eqref{eq:proportions}}}\geq 0.
$$

\begin{remark}
It may seem absurd to introduce the potential $U$ since instead of having only the inequality~\eqref{eq:globalE} to prove, we end up with the inequalities~\eqref{eq:global} and \eqref{eq:local}.
However, the inequality~\eqref{eq:local} is {\em local} and can - at least theoretically - be checked by computer.
Further, the potential will be chosen to reduce the verification of the global inequality~\eqref{eq:global} to the verification of the inequalities~\eqref{eq:edge} and \eqref{eq:vertex}, which are as well local.
In other words, a hard global inequality is replaced by three easier local inequalities (if one can find such a potential $U$).
\end{remark}

Let us now define the potential $U$ and the constants $\alpha_1$ and $\alpha_r$.
Both the edge potential $U_e$ and the vertex potential $U_v$ will be defined as in \cite{BF20}.
We shall here only define the values of the {\em base vertex potentials} $V_{ijk}$ and of the parameters $q_{xy}$ and $l_{xy}$.

The $6$ {\em base vertex potentials} $V_{111}$, $V_{11r}$, $V_{1r1}$, $V_{1rr}$, $V_{r1r}$ and $V_{rrr}$, as well as the real number $\alpha_1$ and $\alpha_r$ are defined by the eight following independent equations.
The first four equations ensure that the sum of the base vertex potentials in any triangle is equal to the emptiness of this triangle (as in \cite{BF20}):
\begin{eqnarray*}
3V_{111}&=&E_{111},\\
3V_{rrr}&=&E_{rrr},\\
2V_{11r}+V_{1r1}&=&E_{11r},\\
2V_{1rr}+V_{r1r}&=&E_{1rr}.
\end{eqnarray*}
The two following equations ensure Ineq.~\eqref{eq:vertex} around large and small discs of the central packing depicted in Fig.~\ref{fig:packings_known}:
\begin{eqnarray*}
8V_{11r}&=&\alpha_1,\\
4V_{1r1}&=&\alpha_r.
\end{eqnarray*}

The seventh equation depends on $x$:
\begin{eqnarray*}
6V_{qqq}&=&\alpha_q,
\end{eqnarray*}
with $q=r$ if $x<\tfrac{1}{2}$ or $q=1$ otherwise.
It ensures Ineq.~\eqref{eq:vertex} around small (resp. large) discs in the leftmost (resp. rightmost) packing in Fig.~\ref{fig:packings_known}.

The eighth and last equation also depends on $x$.
It assigns an arbitrary value to $V_{1rr}$ (the way this value has been chosen is discussed at the end of this section):
$$
V_{1rr}=\left\{\begin{array}{cl}
\tfrac{7x^2+6x-1}{1000} & \textrm{if $x<\tfrac{1}{2}$},\\[10pt]
-\tfrac{9}{1000} & \textrm{otherwise}.
\end{array}\right.
$$

The potential $V_{1rr}$ turns out to be discontinuous in $x=\tfrac{1}{2}$: this is because the seventh equation is completly different depending on whether $x<\tfrac{1}{2}$ or not (this also explains the discontinuities in Fig.~\ref{fig:stats_triangles} and \ref{fig:stats_a1}).

\begin{lemma}
For $0\leq x\leq 1$, the equation~\eqref{eq:proportions} is satisfied.
\end{lemma}

\begin{proof}
First, consider the case $x\geq\tfrac{1}{2}$.
The density $\delta(x)$ is the density of a packing that corresponds to a square-triangle tiling with $\alpha=\tfrac{1-x}{x}$ square for $2(1-\alpha)$ triangles.
Each square corresponds in the FM-triangulation to $4$ triangles, each with two large discs and one small disc, and each triangle corresponds to one triangle with three large discs.
Since the total excess of this packing is zero, one has
$$
0=4\alpha E_{11r}+(1-\alpha)E_{111}=4(1-x)E_{11r}+2(2x-1)E_{111}.
$$
We then rely on the equations that define the $V_{ijk}$'s and $\alpha_q$'s:
\begin{eqnarray*}
0
&=&4(1-x)(2V_{11r}+V_{1r1})+2(2x-1)3V_{111}\\
&=&(1-x)8V_{11r}+(1-x)4V_{1r1}+(2x-1)6V_{111}\\
&=&(1-x)\alpha_1+(1-x)\alpha_r+(2x-1)\alpha_1\\
&=&x\alpha_1+(1-x)\alpha_r.
\end{eqnarray*}

Now, consider the case $x<\tfrac{1}{2}$.
The density $\delta(x)$ is the density of a packing whose FM-triangulation has $4\alpha$ triangles with two large discs and one small disc for $2(1-\alpha)$ triangles with three small squares (and a negligible proportion of other triangles), with $\alpha=\tfrac{x}{1-x}$.
We then proceed as above:
\begin{eqnarray*}
0
&=&4\alpha E_{11r}+2(1-\alpha)E_{rrr}\\
&=&4x E_{11r}+2(1-2x)E_{rrr}\\
&=&4x(2V_{11r}+V_{1r1})+2(1-2x)3V_{rrr}\\
&=&x8V_{11r}+x4V_{1r1}+(1-2x)6V_{rrr}\\
&=&x\alpha_1+x\alpha_r+(1-2x)\alpha_r\\
&=&x\alpha_1+(1-x)\alpha_r.
\end{eqnarray*}
The equation~\eqref{eq:proportions} is thus satisfied for any $x$.
\end{proof}

We further proceed as in \cite{BF20}.
First, one finds coefficients $m_1$ and $m_r$ of the angle deviation in the vertex potential which ensures Ineq.~\eqref{eq:vertex} around any vertex of any FM-triangulation.
We then compute the ceiling $Z_1$ and $Z_r$.
The vertex potential is fully defined.
The edge potential is defined by the parameters given in Tab.~\ref{tab:edge_parameters}.
One then finds $\varepsilon>0$ such that the local inequality~\eqref{eq:local} follows from the mean value theorem for any triangle of the FM-triangulation with distance at most $\varepsilon$ between any two of its discs.
We finally check with a computer program the local inequality~\eqref{eq:local} on all the other possible triangles by dichotomy, using interval arithmetic.

\begin{table}
\centering
\begin{tabular}{c|cccccc}
& $l_{11}$ & $q_{11}$ & $l_{1r}$ & $q_{1r}$ & $l_{rr}$ & $q_{rr}$\\
\hline
$x<0.5$ & $2.5$ & $0.38$ & $1.83$ & $0.15$ & $1.18$ & $0.15$\\
$x\geq 0.5$ & $2.5$ & $0.02$ & $1.83$ & $0.05$ & $1.18$ & $0.08$\\
\end{tabular}
\caption{Parameters for the edge potential.}
\label{tab:edge_parameters}
\end{table}

Last but not least, we have to check that $\delta(x)$ is maximal not only for a given proportion $x$, but for any proportion in $[0,1]$.
Once again, we use interval arithmetic: the interval $[0,1]$ is divided into intervals that are small enough to perform the above check with each of these intervals as a value for $x$.

We first wrote a SageMath/Python program to perform the above check.
It works quite well for a given value of $x$ but is too slow to check all the possible proportions.
We thus wrote a C++ program relying on the Boost interval arithmetic library.
Cutting $[0,1]$ in $100$ intervals of length $0.01$ appeared to be sufficient to check both the local and global inequalities for any proportion $x$.
Fig.~\ref{fig:stats_triangles} gives an idea of the number of triangles on which, for each interval of proportions, the local inequality has been checked by dichotomy (this is the most time-consuming part of the check).
The complete checking took about $3$ minutes on our laptop, an Intel Core i5-7300U with $4$ cores at $2.60$GHz and $16$ Go RAM.

\begin{figure}[hbt]
\centering
\includegraphics[width=\textwidth]{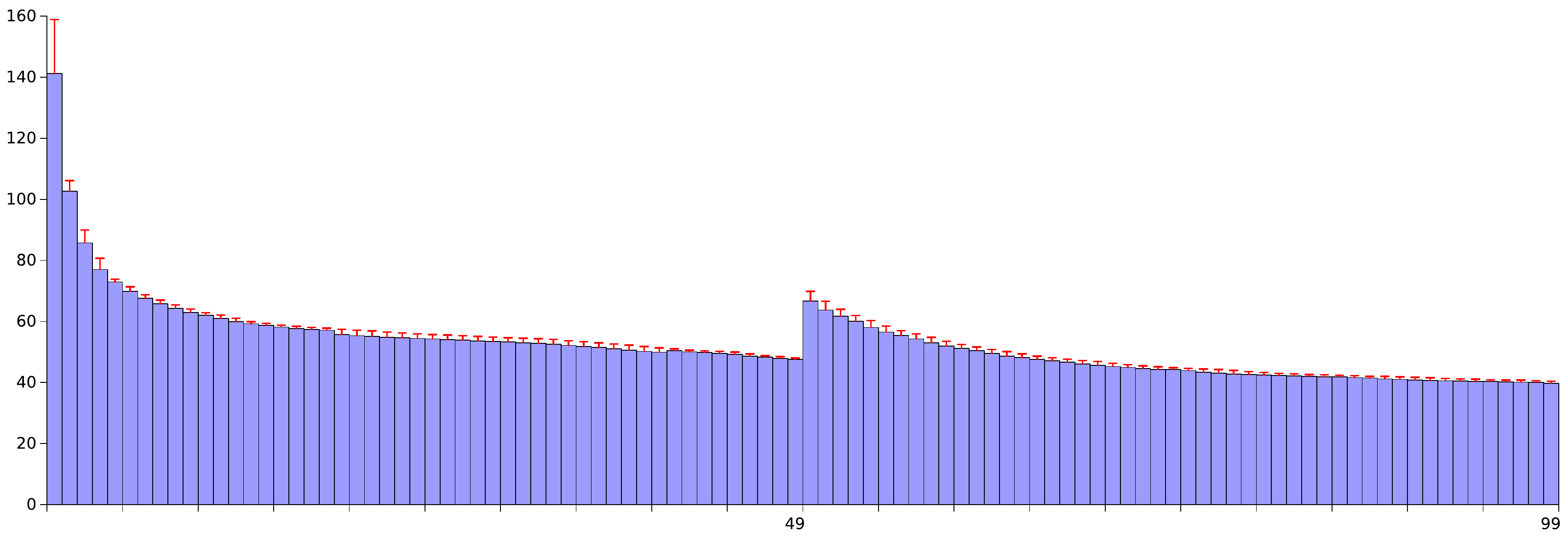}
\caption{
Number of thousands of triangle sets on which the local inequality~\eqref{eq:local} is checked by the computer,  as a function of the percentage of large discs.
The error bars show how this number increases when $\alpha_q$ is replaced by $\alpha_q+\eta$, with $\eta=10^{-4}$, for discs of radius $q\in\{1,r\}$ with a bad neighborhood.
}
\label{fig:stats_triangles}
\end{figure}

It is also worth taking a look at the values of $\alpha_1$ ad $\alpha_r$, depicted in Figure~\ref{fig:stats_alpha}.
For $x<\tfrac{1}{2}$, one has $\alpha_1<0$, that is, the density is locally higher than the average around large discs.
This is because those are involved in the maximum density phase.
On the contrary, $\alpha_r>0$, that is, the density is locally lower than the average around small discs because many of them are involved in the less dense compact hexagonal phases.
The situation is the inverse for $x>\tfrac{1}{2}$.

\begin{figure}[hbt]
\centering
\includegraphics[width=\textwidth]{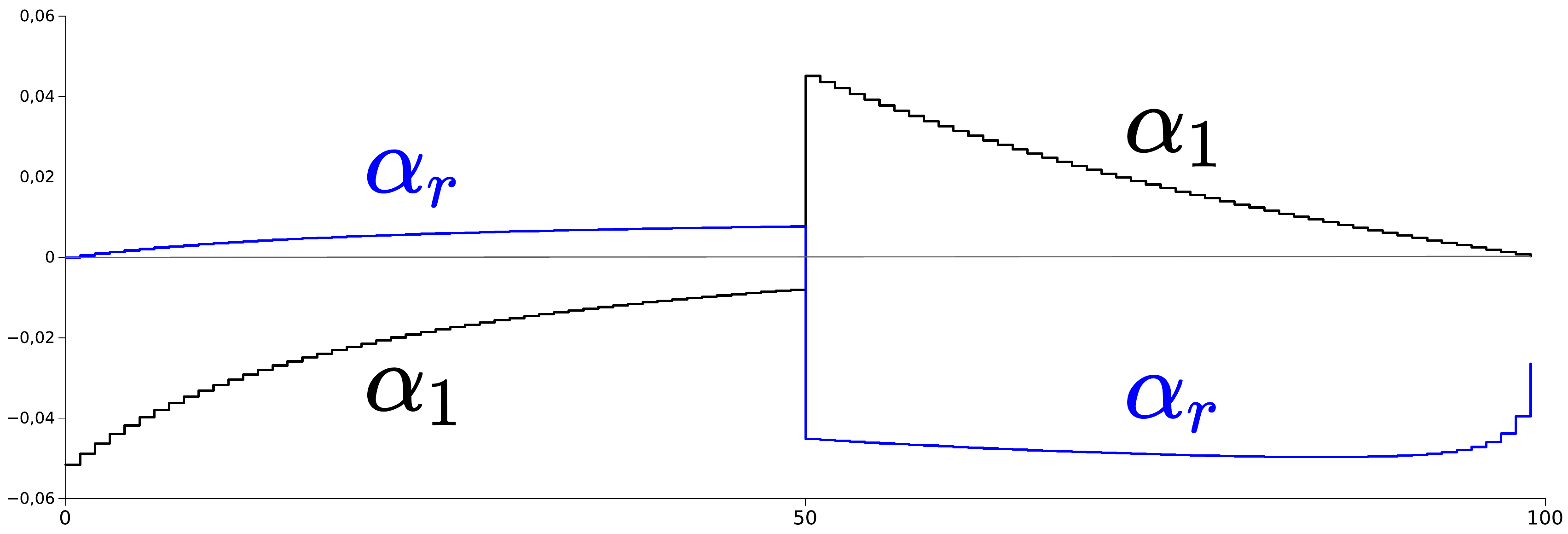}
\caption{
Value of $\alpha_1$ and $\alpha_r$ found by computer to satisfy the vertex inequality~\eqref{eq:vertex}, as a function of the percentage of large discs.
}
\label{fig:stats_alpha}
\end{figure}

To conclude this section, let us explain how the values of $V_{1rr}$ have been chosen.
For a given proportion $x$, we tried different values of $V_{1rr}$.
Some values yield an error ({\em i.e.}, an FM-triangle which does not satisfy the local inequality) and some other values yield an infinite loop-recursion when trying to refine too far the precision while checking the local inequality.
But many values allow to successfully check the local inequality: we took for $V_{1rr}$ the average of these values.
We do this for different proportions $x$, then interpolate the obtained values by a polynomial to define $V_{1rr}$ for any $x$.

%%%%%%%%%%%%%%%%%%%%%%%%
\bibliographystyle{alpha}
\bibliography{proportions}

\end{document}